\documentclass[10pt,conference,compsocconf,letterpaper]{IEEEtran}
\usepackage{times}
\usepackage[english]{babel}
\usepackage{amsmath,amssymb,bbm}
\usepackage{todonotes}
 \usepackage[linesnumbered,ruled,vlined]{algorithm2e}
 \usepackage{amsthm}
      \usepackage{graphicx}
      \usepackage{subfigure}
      \usepackage{epstopdf}
      \usepackage{cite}
\newtheorem{Theorem}{Theorem}
\newcommand{\nin}{\not\in}
\newcommand{\nobracket}{}
\newcommand{\tmfolded}[2]{\trivlist{\item[$\bullet$]\mbox{}#1}}
\newcommand{\tmop}[1]{\ensuremath{\operatorname{#1}}}

\newtheorem{prop}{Proposition}
\newtheorem{lemma}{Lemma}
\newtheorem{definition}{Definition}
\newtheorem{remark}{Remark}

\ifCLASSINFOpdf
\else
\fi

\begin{document}
%
% paper title5
\title{Task-Cloning Algorithms in a MapReduce Cluster with Competitive Performance Bounds}

\author{Huanle XU, Wing Cheong LAU \\ Department of Information Engineering, The Chinese University of Hong Kong\\%, Shatin, N.T., Hong Kong\\
\{xh112, wclau\}@ie.cuhk.edu.hk \\
}

\maketitle

%\vspace{-0.5cm}

\begin{abstract}
Job scheduling for a MapReduce cluster has been an active research topic in recent years. However, measurement traces from  real-world production environment show that the duration of tasks within a job vary widely.
The overall elapsed time of a job, i.e. the so-called flowtime,  is often dictated by one or few slowly-running tasks within a job, generally referred as the  ``stragglers''.  The cause of stragglers include tasks running on partially/intermittently failing machines or the existence of some  localized resource bottleneck(s) within a MapReduce cluster.
To tackle this online job scheduling challenge, we adopt the task cloning approach and design the corresponding scheduling algorithms which aim at minimizing the weighted sum of job flowtimes in a MapReduce cluster based on the Shortest Remaining Processing Time scheduler (SRPT). To be more specific, we first design a 2-competitive offline algorithm
when the variance of task-duration is negligible. We then extend  this offline algorithm to yield the so-called  SRPTMS+C algorithm for the online case and show that SRPTMS+C is $(1+\epsilon)-speed$ $o(\frac{1}{\epsilon^2})-competitive$ in reducing the weighted sum of job flowtimes within a cluster. Both of the algorithms explicitly consider the precedence constraints between the two phases within the MapReduce framework.
We also demonstrate via trace-driven simulations  that SRPTMS+C can significantly reduce the weighted/unweighted sum of job flowtimes by cutting down the elapsed time of small jobs substantially.
In particular, SRPTMS+C beats the Microsoft Mantri scheme by nearly 25\% according to this metric.
\end{abstract}


\section{Introduction}
\label{Introduction}
MapReduce \cite{mapreduce:google} and its open-source realization via Hadoop \cite{hadoop} have emerged as the defacto framework  to support large-scale parallel/distributed processing and data analytics.
Under the  MapReduce framework, the overall computation of a job  is decomposed into 2 separate phases, namely, the Map phase and the Reduce phase. Within each phase, many relatively small tasks are executed in parallel across a large number of  machines within the MapReduce cluster. The MapReduce computational model also requires that the Reduce phase of a job cannot begin until all the tasks within its Map phase have been completed.  A key feature of catalyzing the widespread adoption of MapReduce framework is the ability to transparently deal with the challenges
of executing these tasks in a distributed setting.  One of such fundamental challenges is the disproportionately long-running tasks, or the so called stragglers, which corresponding to tasks that are unfortunately assigned to machines suffering from partially/intermittently failures or localized resource bottleneck(s).  Measurement traces from the real-world production environment   \cite{Outliers} indicate that stragglers lead to a large variation in completion times among tasks in the same job phase and delay job completion 
substantially.
%a job can be delayed substantially by one or few stragglers \cite{Outliers}. In addition, \cite{Outliers} points out the stragglers leads to a large variation in completion times among functionally similar tasks in the same job phase.

The dominant technique to mitigate the straggler problem is via speculative execution: a strategy which preventively or
reactively handle stragglers via automatically launching of extra copies of a task on alternative machines.
%As such, these models do not consider the non-deterministic nature of task durations and are not applicable to real clusters.
In particular, there are two main classes
of speculative execution strategies proposed in the literature, namely, the Cloning approach  \cite{Cloning} and the Straggler-Detection-based one  \cite{hadoop,Dryad,Performance,Outliers,Smart_Speculative}.
Under the Cloning approach, extra copies of
a task are scheduled in parallel with the initial task and the one which finishes first is used for the subsequent computation. For the Straggler-Detection based
approach, the progress of each task is monitored by the
system and backup copies are launched when a straggler is
detected. Unfortunately,
most of these speculative execution schemes are based
on simple heuristics and generally lack any performance guarantee.

%\todo[inline]{needs further modificationi}

To take a more systematic
approach for the design of speculative execution strategies,
our previous work (e.g., \cite{speculative-single,speculative-single-optimization,speculative-multiple-optimization}) propose several optimization-based schemes:
 \cite{speculative-multiple-optimization} proposes to make clones for each task of the arriving jobs by running a convex program which aims at minimizing
the total job elapsed time, which is the time-span between the job arrival and its completion. This is commonly referred as the flowtime of a job in the scheduling literature. However, \cite{speculative-single,speculative-single-optimization,speculative-multiple-optimization} still
 face two fundamental limitations.   Firstly, the precedence constraints between the two phases in the MapReduce framework are ignored.
 Secondly, the complete distribution of task duration within each job needed to be known in advance when solving the optimization problem. 
 Ideally, we want to take the precedence
 constraint into consideration and reduce the amount of information required for optimizing the speculative execution scheme. 
 %However, all of these proposed strategies are either heuristic based or need to estimate the As such, \cite{speculative-single,speculative-single-optimization,speculative-multiple-optimization} To ake a more systematic, optimization-based
%approach for the design of speculative execution strategies along with the scheduling algorithm. However, these models do not consider the phase dependency constraint and assume that the complete
%distribution information of task duration is known \textit{a priori}, which is quite difficult to achieve in reality.

With the above ideas in mind, in this paper, we explicitly model the precedence between the Map and Reduce phase and assume that only the first and second moments of task duration are known \textit{a priori}. Similar to \cite{speculative-multiple-optimization}, we aim to minimize the weighted sum of  job flowtime via task cloning. This objective yields 
offline and online versions of the scheduling problem which turns out to be more difficult than the NP-Hard scheduling problem presented in \cite{Schedulers}. Our main results include the approximated algorithms which are motivated from the Shortest Remaining Processing Time scheduler (SRPT) in both offline and online setting. To be more specific, we obtain a 2-competitive algorithm for the offline case
when the variance of task-duration is negligible and a $(1+\epsilon)-speed$ $o(\frac{1}{\epsilon^2})-competitive$ algorithm for the online case where $0 < \epsilon < 1$. For the online version of the algorithm, we assume resource augmentation \cite{speed}, which is necessary to circumvent lower bounds for the parallel scheduling on multiple machines. Under the resource augmentation analysis, the adversary is given $m$ unit-speed machines and our algorithm is given $M$ processors of speed $s$ where $s > 1$.
%we relax the assumption in \cite{speculative-single,speculative-single-optimization,speculative-multiple-optimization} and take the phase dependency constraint into consideration in the meanwhile. More precisely, we assume that only the expectation and variance  of task duration within each phase is known in advance and task cloning can speed up job processing.  We then formulate a stochastic optimization program whose objective is to minimize the expectation of total job delay, i.e., the job flowtime which is defined as follows:
%For solving this problem, we design two approximated algorithms which are extended from the SRPT principle in the offline and online cases respectively. Moreover, we have provided a theoretical guarantee for these two algorithms. To demonstrate the efficiency of our proposed online algorithm, we run several trace-driven simulations to compare the total weighted job flowtime with other three baseline algorithms. It shows that our algorithm reduces such metric by nearly 25\% comparing to Microsoft Mantri baseline scheme which adopts the straggler-detection-based speculative execution scheme.
To summarize, this paper has made the following technical contributions:

\begin{itemize}
\item After reviewing the related work in Section \ref{related_work}, we cast the dynamic scheduling problem as an stochastic optimization problem that focuses on finding a cloning scheme to minimize the weighted sum of job flowtimes (Section \ref{system_model}).
\item Motivated by the SRPT scheduler, we design a 2-competitive algorithm for the offline case when the variance of task duration is negligible. Moreover, we show that,
with high probability, each job can complete within a time-span which is larger than the optimal algorithm by only a constant factor times the standard derivation of task duration (Section \ref{bulk_arrival}).
\item Extended from the offline algorithm, we design the so-called SRPTMS+C algorithm for the online case.
By adopting the method of potential function analysis,
%\cite{competitive,energy_efficient,scalably-scheduling,SRPT_identical}
we prove that SRPTMS+C is $(1+\epsilon)-speed$ $o(\frac{1}{\epsilon^2})-competitive$ for the weighted sum of job
flowtimes when $0< \epsilon < 1$ (Section \ref{online-scheduling}).
\item Before concluding our work in Section \ref{conclusion},  we demonstrate via trace-driven simulations that SRPTMS+C can significantly reduce the weighted average of job flowtimes by cutting down the elapsed time of small jobs substantially. In particular, SRPTMS+C beats the Microsoft Mantri scheme by nearly 25\% according to this metric (Section \ref{evaluation}). 
\end{itemize}

\section{Related work}
\label{related_work}
The straggler problem was first identified in the original MapReduce paper \cite{mapreduce:google}. Since then, various solutions have been proposed to deal with it using the
Straggler-Detection-based  speculative execution strategy \cite{Dryad,Performance,Outliers,Smart_Speculative}. These solutions mainly focus on promptly identifying
 stragglers and accurately predicting the performance of running tasks.  One fundamental limitation is that detection may be too late for helping small jobs as it needs to wait for the collection of enough samples while monitoring the progress of tasks. To avoid the extra delay caused by the straggler detection, cloning approach was proposed in \cite{Cloning}. This approach
 relies on cloning very small job in a greedy manner to mitigate the straggler-effect and is based on simple heuristics. In contrast, we develop an optimization framework to make clones for each arriving job. Recently, \cite{grass} presents GRASS, which carefully adopts the Detection-based approach to trim stragglers for approximation jobs. GRASS also provides a unified solution for normal jobs. However, one limitation is that it only prioritizes the tasks within a job and it remains a problem to prioritize different jobs (i.e., the scheduler is not optimized and unknown to the readers).

Prior research on job scheduling for a MapReduce system includes \cite{Fast_completion,Flow_Shops,Joint_Phase,Delay_Tails,Joing_scheduling,overlapping_phases,Schedulers}:\cite{Fast_completion,Joing_scheduling,Joint_Phase}
derive performance bounds for minimizing the total completion time. \cite{Delay_Tails} designs the \textit{Coupling scheduler}, which mitigates the starvation problem caused by Reduce tasks in large jobs. \cite{Flow_Shops,Schedulers,Joing_scheduling} extend the SRPT scheduler to minimize the total job flowtime under different settings. However, all of these studies assume accurate knowledge of task durations and hence do not support speculative copies to be scheduled dynamically.

Finally, the SRPT scheduler has been studied extensively in traditional parallel scheduling literature.
%a bunch of work to use a \textit{resource augmentation
%analysis} \cite{SRPT_identical}.
In particular, SRPT has proven to be $(1+\epsilon)-speed$ $\frac{4}{\epsilon}-competitive$ for total flowtime on $m$ identical machines under the single task case \cite{SRPT}.
In this paper, we extend the SRPT scheduler to yield an online scheduler which can mitigate stragglers as well.

\section{System Model and Problem Formulation}
\label{system_model}
Consider a MapReduce Cluster which consists of M machines. A machine could represent a processor, a core or a virtual machine. Assume a set of jobs  $\mathcal{J} = \{J_1, J_2, \cdots\}$ entering into the cluster over time. Job $J_i \in \mathcal{J}$ which arrives at the cluster at time $a_i$ consists of $m_i$ map tasks and $r_i$ reduce tasks. Each job has a weight $w_i$ which reflects its priority. Let $J_i^m = \{\delta_i^{m,1},\delta_i^{m,2},\cdots,\delta_i^{m,m_i}\}$ and $J_i^r = \{\delta_i^{r,1},\delta_i^{r,2},\cdots,\delta_i^{r,r_i}\}$ be the set of map and reduce tasks of $J_i$ respectively. Each machine can only hold one map or reduce task at any time and all the machines are identical.

As described in Section \ref{Introduction}, the large variation in task completion times is caused by machine failures or localized resource bottleneck(s). Instead of modelling the variance of machine speed directly, we consider that the variation is caused by task workload and each machine processes all the tasks with the same speed. Such transformation does not violate the variation in task completion times and simply our analysis.

We assume time is slotted and a centralized scheduler collects the status of jobs within the cluster at the beginning of each time slot. If a machine runs a task at speed $s$, it will take $p(\cdot)/s$ time slots to complete the task where $p(\cdot)$ denotes the workload of this task.
 Without loss of generality, we assume that  all the machines run at unit speed.

For ease of presentation, throughout the whole paper, we use $c\in \{m,r\}$ to capture the map- or reduce-related statements for all the tasks, i.e., when $c$ is used, it is fixed to either $m$ or $r$. The workload of task $\delta_i^{c,j} \in J_i^c$ is $p_i^{c,j}$ where $p_i^{c,j}$ is a random number for all $i,j$. Under the unit speed case, $p_i^{c,j}$ also denotes the processing time of task $\delta_i^{c,j}$ on a particular machine. We also assume the workload of all tasks in a job share the same mean $E^c_i$ and standard deviation $\sigma^c_i$. The parameters $E^c_i$ and $\sigma^c_i$ are known in advance to the scheduler for all $i$.

Table \ref{Table_1} summarizes all the notations in this model.

\begin{table}
%\vspace{.5em}
\centering
\caption{The notations of the scheduling parameters}
\label{Table_1}
\begin{tabular}{|c|c|}
  \hline
  Notations &  Corresponding meaning\\
  \hline
  \hline
  $\mathcal{J}$ & The set of jobs arriving at the cluster\\
  \hline
  $J_i^c$ & The set of map/reduce tasks of job $J_i$ with \\
  & $c=\{map$ for map task; $reduce$ for reduce task\}\\
  \hline
  $a_i$ & Arrival time of job $J_i$\\
  \hline
  $f_i$ & Time when job $J_i$  completes its work \\
  \hline
  $w_{i}$ & Weight of job $J_i$\\
  \hline
  $p_i^{c,j}$ & Workload of the map/reduce task $\delta_i^j$($\rho_i^j$)\\
  \hline
  $E^c_i$ & The mean of the workload for map/reduce task in $J_i$\\
  \hline
  $\sigma^c_i$ & The SD of the workload for map/reduce task in $J_i$\\
  \hline
  M & Total number of machines in the Cluster\\
  \hline
  $s^c_i(x)$  & The speedup function of a map/reduce task in $J_i$\\
  \hline
  $c^{j}_{i}$  & Time when task $\delta^{{c,j}}_{i}$ is scheduled.\\
  \hline
   $t^{c,j}_{i}$  & The duration of task $\delta^{{c,j}}_{i}$.\\
  \hline
  $f^{c,j}_{i}$ & Time when task $\delta^{{c,j}}_{i}$ completes.\\
  \hline
  M(t) & Number of machines running map tasks at time $t$.\\
  \hline
  R(t) & Number of machines running reduce tasks at time $t$.\\
  \hline
  $x^{c,j}_{i}$ & Number of copies made for task $\delta^{{c,j}}_{i}$.\\
  \hline
\end{tabular}
%\vspace{-1.3em}
\end{table}

\subsection{Speedup via task cloning}
In this model, we adopt the cloning approach to mitigate the negative 
impact of stragglers. Cloning helps to speed up the completion of a task via picking up the copy which finishes first of this task.
To capture such speedup, we define a function, which is $s^c_i(x)$,  for each phase of every single job where $x$ is the number of copies made for a particular task. For example, it takes $p_i^{c,j}/s^c_i(2)$ time slots to complete task $\delta_i^{c,j}$ on average if two copies are made when scheduling $\delta_i^{c,j}$. Here, we assume that $s^c_i(x)$ satisfies the following two properties:
\begin{itemize}
\item $s^c_i(x)$ is a concave and strictly increasing function of $x$, $\forall i$.
\item $s^c_i(1)=1$ and $s^c_i(x) \leq x$ for all $x > 0$, $\forall i$.
\end{itemize}
These two properties are applicable to most distributions of the task duration observed in practice. For example, 
\cite{Outliers,speculative-multiple-optimization} show that the task duration for a MapReduce cluster follows a heavy-tail distribution.  Below, we illustrate the convexity of 
the speedup function when the task duration follows a Pareto heavy-tail distribution.  In particular, if the duration $p_i^{c,j}$ of task $\delta_i^{c,j}$ follows the following Pareto distribution, we have:
$$ Pr(p_i^{c,j} < t) = \left\{\begin{array}{cc}
1-(\frac{\mu}{t})^{\alpha} & for \ t \geq \mu\\
0 & otherwise
\end{array}\right.$$
when $r$ copies are made for the task $\delta_i^{c,j}$, the average duration of $\delta_i^{c,j}$ is $\frac{\alpha r \mu}{\alpha r - 1}$. The derivation of this result is shown in \cite{speculative-single-optimization}. As such, the speedup function is just $s_i^c(r) = \frac{r\alpha - 1}{r(\alpha-1)} $ which is strictly concave and
monotonic.

%\vspace{-.5em}
\subsection{A stochastic program formulation for job scheduling}
For any job, all the map tasks and reduce tasks can only be scheduled after the job arrival at the cluster and hence $m^{j}_{i} \geq a_{i}$. The Map phase ends when all the map tasks finish, i.e.,  $f_{i}^{m,j} =m^{j}_{i} + t^{m,j}_i \ \  \forall i;1\leq j \leq m_i$. Due to the precedence constraints of the Map and Reduce phase, a reduce task can not begin its work if some map tasks within the same job do not finish. Thus, the reduce task $\delta_i^{r,j}$ can only start after the end of the Map phase (i.e., $\max \{ \max_{k} \{ f_{i}^{m,k} \} ,r^{j}_{i} \}$). At any time slot, the total number of machines available for processing the tasks and their clones cannot exceed $M$, i.e., $M(t) + R(t) \leq M$. Finally, a job completes when all the reduce tasks are finished, i.e.,  $f_{i} = \max_{j} \{ f_{i}^{r,j} \} \ \ \forall i;1\leq j \leq r_i$.

For this model, we aim to minimize the weighted sum of job flowtimes by carefully making cloning decisions and prioritizing different jobs. This formulation yields
an optimization problem shown below:
%\vspace{-.5em}
\begin{subequations}
\label{online_formulation}
\begin{eqnarray}
  \min &  & \sum_{i} w_{i} \cdot \mathbbm{E}[ f_{i} -a_{i} ]\\
  \label{released}
  s.t. &  & m^{j}_{i} \geq a_{i} \quad \forall i;1\leq j \leq m_i\\
  &  & r^{j}_{i} \geq a_{i} \quad \forall i;1\leq j \leq r_i\\
  \label{mapper_processing}
  &  & \mathbbm{E}[t^{m,j}_i] =  {E^{m}_{i}}/ {s^m_i(x^{m,j}_{i})} \quad \forall i;1\leq j \leq m_i\\
  \label{reducer_processing}
  &  & \mathbbm{E}[t^{r,j}_i]  = {E^{r_{}}_{i}}/{s^r_i(x^{r,j}_{i})} \quad \forall i;1\leq j \leq r_i\\
  &  & f_{i}^{m,j} =m^{j}_{i} + t^{m,j}_i \quad \forall i;1\leq j \leq m_i\\
  \label{reduce_begin}
  &  & f_{i}^{r,j}  =  \max \{ \max_{k} \{ f_{i}^{m,k} \} ,r^{j}_{i} \} +  t^{r,j}_i \ \forall i;j\\
  \label{num_map_machine}
  &  & \sum_{{m_{i}^{j} \geq t; \ f_{i}^{m,j} <t}} x^{m,j}_{i} =M(t) \quad \forall t\\
   \label{num_reduce_machine}
  &  & \sum_{{r_{i}^{j} \geq t; \ f_{i}^{r,j} <t}} x^{r,j}_{i} =R(t) \quad \forall t\\
  \label{num_total_machine}
  &  & M( t ) +R( t ) \leq M \quad \forall t\\
  \label{flowtime}
  &  & f_{i} = \max_{j} \{ f_{i}^{r,j} \} \quad \forall i;1\leq j \leq r_i
\end{eqnarray}
\end{subequations}
%The objective is to minimize the expectation of the overall weighted job flowtime.
Constraint (\ref{mapper_processing}) and \eqref{reducer_processing} illustrate the speedup property for the map tasks and reduce tasks respectively.
Constraint \eqref{reduce_begin} is due to the precedence constraint of the Map and Reduce phase.

\begin{remark}
\label{remark_1}
When task cloning is not used and there is no variation in completion times among tasks in the same job phase, the scheduling problem in our model just reduces to the problem in \cite{Schedulers}. However, the optimization problem presented in \cite{Schedulers} has proven to be NP-Hard even for the offline case where all the jobs enter into the cluster at the same time. The stochastic optimization problem in Equation \eqref{online_formulation} therefore is NP-Hard and hence  we resort to the use of approximation algorithms to tackle this problem. 
\end{remark}

\section{Offline Scheduling: All the jobs arrive at the cluster at the same time}
\label{bulk_arrival}
Before designing the online algorithm, in this section, we consider an offline case where all the jobs enter into the system at the same time, namely, $a_i = 0 \ \forall i$.
%All the parameters $E^c_i$ and $\sigma^c_i$ are known to the scheduler in advance.
We assume that all the tasks cannot be launched simultaneously in the cluster (e.g., $\sum_{i=1}c_i > M$), otherwise, we can assign all the tasks to the machines in the cluster and the scheduling process just ends. Although this setting is simple, the offline algorithm presented below provides good insights for us to design an online algorithm in the following sections.

\cite{grass} builds a simple model to analyze the advantage of pure cloning. It concludes that cloning cannot help to reduce the job flowtime if $s^c_i(x) \leq x$ when the number of tasks to be scheduled is larger than $M$. Therefore, we do not clone extra copies for the tasks in this bulk arrival scenario.

%\vspace{-1.5em}

\subsection{Offline Algorithm Design}
It is well known in scheduling literature that the SRPT scheduler is optimal for reducing overall flowtime on a single machine when there is only one task per job \cite{Scheduling_book}. In each time slot, the SRPT scheduler always selects the job with the minimum total remaining workload to serve.  We extend this SRPT scheduler to design the following offline scheduling algorithm:

In this algorithm, the scheduler first applies the SRPT scheduler for the scheduling problem in which there is only one single machine to determine the priority of each job.
Let $\phi_i$ be the total effective workload of job $J_i$, which is determined by the following equation:
\begin{equation}
\label{workload}
\phi_i = m_i\cdot(E^m_i + r\sigma^m_i) + r_i\cdot(E^r_i + r\sigma^r_i)
\end{equation}
The standard deviation of task duration is incorporated into the workload via multiplying by a factor $r$ and the priority of job $J_i$ is then defined as $w_i/\phi_i$. The rationale of including the standard derivation of task duration in the effective workload of task is that tasks with large variation in completion times can easily prolong the job completion and hence should be scheduled later. However, it still
remains a problem to choose a good $r$ and we tackle this problem in Section \ref{evaluation}.

After computing the priority for each job,  the jobs with higher priorities are always scheduled before those ones with lower priorities. Whenever a machine is available, the scheduler randomly chooses one unscheduled task from the pool of not-yet-finished jobs, or the set of alive jobs that keep the highest priority and assign it to this machine. Moreover, all the map tasks are scheduled before the reduce tasks in the same job.
Since the Reduce phase can only begin after the Map phase finishes, a reduce task cannot make progress even after it has been scheduled as long as there are some unfinished map tasks within the same job.

%\vspace{-.5em}
\IncMargin{1em}
\begin{algorithm}
\label{offline_algorithm}
\caption{Offline Scheduling algorithm for the bulk arrival}
\Indm
\KwIn{The jobs associated with $c_i$, $E^c_i$ and $\sigma^c_i$;}
\KwOut{Allocated machines for all the map tasks and reduce tasks.}
\Indp
Sort the job set $\mathcal{J}$ based on the decreasing order of $w_i/\phi_i$ \;
Initialize the job set $\Phi = \mathcal{J}$\;
\If{A machine is  available}
{
\For{each job $J_i$ in $\Phi$ }
{
\eIf{$J_i$ has unscheduled map task}{
Choose one unscheduled map task at random and assign it to this machine\;
}
{
Choose one unscheduled reduce task at random and assign it to this machine\;
}
\If{$J_i$ has no unscheduled task}
{
$\Phi = \Phi - \{J_i\}$\;
}
}
}
\end{algorithm}
\DecMargin{1em}

Algorithm \ref{offline_algorithm} presents the pseudo-code of the algorithm.

\subsection{Deriving the upper bound for job flowtime}
We proceed to analyze the performance of Algorithm \ref{offline_algorithm}. Define $f^s_i$ as the accumulated workload of those jobs whose priority is larger than job $J_i$. In other words,
\begin{equation}
f_i^s = \sum_{j:w_j/\phi_j \geq w_i/\phi_i} \phi_j
\end{equation}
We aim to show a generic bound on the flowtime of each job with a certain probability. To achieve this goal, we first prove the following lemma:
\begin{lemma}
\label{Lemma_1}
With probability at least $\frac{r^2-1}{r^2}$, the cluster is processing the jobs with priority at least $w_i/\phi_i$ during the interval $[0,f_i-E^r_i - r\sigma^r_i].$
\end{lemma}
Refer to Appendix \ref{proof_lemma_1} for the detailed poof of Lemma \ref{Lemma_1}.
Based on Lemma \ref{Lemma_1}, we derive the following theorem which provides an upper bound for the flowtime of
each job:
\begin{Theorem}
\label{Theorem 1}
The flowtime of Job $J_i$ is bounded by $E^r_i + r \sigma^r_i + f_i^s/M$ with a probability at least $1 + 1/r^4 -\frac{2}{r^2}$.
\end{Theorem}
Refer to Appendix \ref{proof_theorem_1} for the proof of Theorem \ref{Theorem 1}.

\begin{remark}
When the variance of the task workload is zero, the flowtime of each job is bounded by $E^r_i + f_i^s/M$ under Algorithm \ref{offline_algorithm}. Regardless of the type of 
scheduler, the flowtime of each job must be larger than $E^r_i$. On the other hand, the performance of the optimal scheduler is no better than the SRPT scheduler with one machine in terms of weighted sum of job flowtimes. Under the SRPT scheduler with one machine, the flowtime of each job is just $f_i^s/M$. Hence, we conclude that Algorithm \ref{offline_algorithm} achieves a constant competitive ratio of two.

However, if the variance is non-negligible, Algorithm \ref{offline_algorithm} could not achieve a constant competitive ratio but still provides an upper bound for the flowtime of each individual job.
\end{remark}

\section{Online scheduling with cloning for job arrival over time}
\label{online-scheduling}
\vspace{-.3em}

In this section, we first present an approximated algorithm for the online scheduling case where all the jobs arrive at the cluster over time. After that, we
provide an upper bound for the competitive ratio of the proposed algorithm.

\subsection{Shortest Remaining Processing Time based Machine Sharing Principle}
\label{algorithm_design}
Extended from the offline algorithm presented in Section \ref{bulk_arrival},  we design a SRPT based machine sharing algorithm for this online case. The principle of machine sharing  is motivated by the work in \cite{scalably-scheduling,Reduce_variance,energy_efficient} where the machines are shared among the latest jobs arriving at the cluster (LAPS).
A classic result in \cite{scalably-scheduling} shows that
the LAPS algorithm  is scalable for minimizing the total flowtime of jobs with sublinear speedup functions on multiple machines. However, in our algorithm, we share the machines among jobs with the smallest remaining workload.  Other than that, we also make clones for the tasks according to machine availability.

We call this approximated algorithm the \textit{Shortest Remaining Processing Time based Machine Sharing plus Cloning} (SRPTMS+C).
At a high level, the SRPTMS+C algorithm works as follows: At the beginning of each time slot, the scheduler computes a priority for every alive job (i.e., not-yet-finished job). Let $\epsilon$ be a number such that $0 < \epsilon < 1$.  Jobs with the highest priorities share the machines in proportion to their weights so that the weight of all running jobs is an $\epsilon$ fraction of the total weights of  the alive jobs in the system. Observe that when $\epsilon$ is set to 1, the scheduler just reduces to the fair scheduler in  Hadoop \cite{hadoop}. On the other hand,  if $\epsilon$ is close to 0, the scheduler becomes the SRPT scheduler. By tuning the parameter $\epsilon$, we could obtain a scheduler
that best fits a cluster.  More importantly, this $\epsilon$ fraction sharing principle yields a bounded competitive ratio as presented in the following sections.

Let $\psi^s (l)$ be the set of alive jobs at the beginning of time slot $l$. Denote by $m_i(l)$ and $r_i(l)$ the number of unscheduled map and reduce tasks of job $J_i$ respectively. The remaining effective workload of job $J_i$ can be characterized by:
\begin{equation}
U_i(l) = m_i(l)\cdot(E^m_i + r\sigma^m_i) + r_i(l)\cdot(E^r_i + r\sigma^r_i)
\end{equation}
The scheduler computes $\frac{w_i}{U_i(l)}$ for each job in $\psi^s (l)$ and guarantees that the jobs with larger $\frac{w_i}{U_i(l)}$ have higher priorities to be scheduled. Let
%\vspace{-.5em}
\begin{equation}
\label{sum-weight}
W(l) = \sum_{i \in \psi^s (l)}w_i
\end{equation}
%\vspace{-.1em}
and $\psi_i^s (l)$ be the set of jobs alive for SRPTMS+C at time slot $l$ which have lower priorities (i.e., smaller $\frac{w_i}{U_i(l)}$) than $J_i$. $J_i$ is also included in $\psi_i^s (l)$. Define $W_i(l) = \sum_{j \in \psi_i^s (l)}w_i$
and let
 $$g_i(l) = \left\{\begin{array}{cc}
\frac{w_i \cdot M}{(\epsilon W(l))} & W_i(l) - w_i \geq (1-\epsilon)W(l) \\
0   & W_i(l) < (1-\epsilon)W(l) \\
\frac{(W_i(l) - (1-\epsilon)W_l)\cdot M}{(\epsilon W(l))} & otherwise
\end{array}\right.$$

Each job $J_i \in \psi^s (l)$ shares $g_i(l)$ machines within the cluster, including those ones that are still running the tasks of $J_i$, whose size is defined as $\sigma_i(l)$. Hence, the number of machines assigned  to $J_i$ in time slot $l$ is $(g_i(l) - \sigma_i(l))$.

\IncMargin{1em}
\begin{algorithm}
%\vspace{1em}
\label{ESE_code}
%\vspace{1em}
\caption{SRPTMS+C Algorithm Design for Online Scheduling}
\Indm
%\KwIn{The alive jobs in the cluster associated with their weights at time slot $l$ and running status of each job and machine;}
%\KwOut{Scheduling decisions for time slot $l$.}
\Indp
Update $\psi^s (l)$,  the set of jobs which have unscheduled tasks at current time slot;

Update the number of available machines $M(l)$;

Compute ${U_i(l)}$ for each $J_i \in \psi^s (l)$ and sort the jobs according to the decreasing order of $\frac{w_i}{U_i(l)}$;

Compute $W(l)$ based on Equation \eqref{sum-weight};

\For{the Job $J_i \in \psi^s (l)$}
{
Compute $g_i(l)$, the number of machines $J_i$ deserved according to the $\epsilon$ fractional sharing policy;
}

\For{the Job $J_i \in \psi^s (l)$ \&\& $g_i(l) > 0$}
{
\label{repeat-process}

Count the number of machines which still run the tasks of $J_i$ including all the clones and denote it by $\sigma_i(l)$;

%\If{$g_i(l) == 0$}
%{
%break;
%}

Compute the number of newly available machines which is $\xi_i(l) = g_i(l) - \sigma_i(l)$;

\If{$\xi_i(l) \leq 0$}
{
 continue;
}
\If{$\xi_i(l) < M(l)$}
{
 Assign $\xi_i(l)$ extra machines to $J_i$;

\textbf{Call} the task scheduling procedure for $J_i$ with $\xi_i(l)$ machines with returning value $\pi_i(l)$;

 $M(l)$ -= $\pi_i(l)$;
}
\If{$\xi_i(l) \geq M(l)$}
{
 Assign $M(l)$ extra machines to $J_i$;

 \textbf{Call} the task scheduling procedure for $J_i$ with $M(l)$ machines with returning value $\pi_i(l)$;

 $M(l)$ -= $\pi_i(l)$;

 %\Return;
}
}

%Update $\psi^s (l) = \psi^s (l) - \{J_i: g_i(l) >0 \}$;

%\If {$M(l) > 0$  \&\& $\psi^s (l) \neq \emptyset$ }
%{
%\For{the Job $J_i \in \psi^s (l)$}
%{
%Compute $g_i(l)$, the number of machines $J_i$ deserved according to the $\epsilon$ fractional sharing policy with $W(l)$ and  $M(l)$ available machines;
%}
%Goto \ref{repeat-process};
%}

return;
\end{algorithm}
\DecMargin{1em}
%\vspace{-.3em}

\IncMargin{1em}
\begin{procedure}
\caption{Task Scheduling for Job $J_i$ with $x$ newly allocated machines ()}
\Indm
\KwIn{The number of newly allocated machines $x$ and the running status;}
\KwOut{Task scheduling decision for $J_i$ and returning value $\pi_i(l)$.}
\Indp
Count $m_i(l)$ and $r_i(l)$, the number of unscheduled map tasks and reduce tasks for $J_i$ respectively;

\uIf{$m_i(l) > 0$ \&\& $m_i(l) \geq x$}
{
 run $[x/m_i(l)]$ copies for each unscheduled task on available machines.

 \Return $x - [x/m_i(l)]*m_i(l)$;
}
\uElseIf{$m_i(l) > 0$ \&\& $m_i(l) < x$}
{
 Choose $x$ unscheduled map tasks uniformly at random and run one copy for each task on available machines;

 \Return 0;
}
\Else
{
Repeat the same scheduling process for reduce tasks with $x$ allocated machines.
}

\label{task-mapping}
\end{procedure}
%\vspace{-1em}
\DecMargin{1em}
%\vspace{-1.2em}

 \subsection{Task-Cloning Algorithm Design}
When allocating the number of machines for each job (i.e., $(g_i(l) - \sigma_i(l))$), there may exist one case which violates the basic sharing principle in
Section \ref{algorithm_design}, namely, the number of machines running the tasks of $J_i$ (i.e., $\sigma_i(l)$) already exceeds $g_i(l)$ for some $i$.
Under such situation, the scheduler reserves the work already completed for job $J_i$ and just runs the tasks of $J_i$ with their clones on $\sigma_i(l)$ machines. 
In other words,
the scheduler does not allow preemption and lets $J_i$ occupy these extra machines.
Due to this non-preemptive mechanism, the exact number of machines shared by Job $J_i$ may be larger than $g_i(l)$.

After the number of machines is allocated for each job, the scheduler needs to choose appropriate tasks of the alive jobs for scheduling and make cloning decisions carefully. Following the precedence  constraint of the Map and Reduce phase, the scheduler begins to schedule reduce tasks after all the map tasks completed. In addition, the clones are made for the tasks depending on whether the number of machines allocated to a particular job is larger than the number of unscheduled tasks.  Take job $J_i$ for example: When $g_i(l) - \sigma_i(l) > c_i(l)$, cloning will be made to fully utilize these machines allocated  to $J_i$. To be more specific, the scheduler spawns the same number of clones for all the unscheduled tasks in $g_i(l)$. Otherwise, tasks with fewer clones are more likely to lag behind. Thus, each unscheduled task of $J_i$ will be made $[(g_i(l) - \sigma_i(l))/c_i(l)]$ \footnote{$[x]$ denotes the rounding of the real number $x$.} copies. 
In contrast, when $g_i(l) - \sigma_i(l) \leq c_i(l)$, following the same argument of the offline scheduling algorithm, clones are not made in this case. Hence, the scheduler chooses some unscheduled tasks from $J_i(l)$ at random and launch it without cloning.  

Algorithm \ref{ESE_code} presents the pseudo-code of the algorithm.

%\vspace{-.6em}
\subsection{Resource augmentation analysis}
In this section, we use resource augmentation to analyze the performance of the SRPTMS+C algorithm.
Before going to the details of the analysis,
we first present the following definition which characterizes the performance of an approximated algorithm.

%\vspace{-.3em}
\begin{definition}
An approximated algorithm is \textit{s-speed c-competitive} if the algorithm's objective is within a factor of $c$
of the optimal solution's objective when the algorithm is given $s$ resource augmentation \cite{speedis}.
\end{definition}

%\vspace{-1.0em}
\begin{prop}
\label{convex_function}
Consider any continuous and concave function $f: \mathbb{R}^{+} \rightarrow \mathbb{R}^{+}$ with $f(0) \geq 0$. Then for any $b \geq a > 0$, we have
 $\frac{f(a)}{a} \geq \frac{f(b)}{b}$.
\end{prop}

%\vspace{-.6em}
\begin{proof}
According to the definition of concave function, it holds that $f(\lambda x + (1-\lambda)y) \geq \lambda f(x) + (1-\lambda)f(y)$ for any
$x,y \in \mathbb{R}^{+}$ and $\lambda \in [0,1]$. Specially, consider $x=0$, $y=b$ and $\lambda = 1 - \frac{a}{b}$. Then we have $f(a) \geq (1-\frac{a}{b})f(0) +
\frac{a}{b}f(b) \geq \frac{a}{b}f(b)$. Q.E.D.
\end{proof}

\begin{remark}
\label{extention}
Based on Proposition \ref{convex_function}, we conclude that $f(\frac{1}{\xi}\cdot a) \geq \frac{1}{\xi} f(a)$, this can be proved by substituting $x = \frac{1}{\xi}\cdot a$ and
$y = a$ into the inequality.
\end{remark}

%\vspace{-.3em}
With the help of Proposition \ref{convex_function}, we derive the following theorem which provides an upper bound
for the competitive ratio of SRPTMS+C.

%\vspace{-.5em}
\begin{Theorem}
\label{competitive_ratio}
The algorithm SRPTMS+C is $(1+\epsilon)-speed$ $o(\frac{1}{\epsilon^2})-competitive$ for the expectation of weighted sum of job flowtimes
when $0< \epsilon < 1$.
\end{Theorem}
%\vspace{-.3em}
The method of potential function analysis is widely adopted to derive performance bound with resource augmentation for
online parallel scheduling algorithms in the literature (e.g., \cite{competitive,energy_efficient,scalably-scheduling,SRPT_identical}).
The key step of this method is to define a proper potential function which combines the adversary and our algorithm.
To be more specific, let $A(t)$ and $OPT(t)$ be the accumulated weighted sum of job flowtimes in the algorithm's and adversary's
schedules, respectively. We define a potential function $\Phi(t)$ that satisfies the following properties which are extended from \cite{competitive}:

\begin{itemize}
\item Boundary Condition: $\Phi(0) = \Phi(\infty) = 0$.
\item Changes Condition when job arrives or completes: the value of the potential function decreases or remains the same  when a job arrives or completes
in our algorithm and the adversary.
\item Dynamic Changes Condition: with $\epsilon$ resource augmentation,
at any time when no job arrives or completes, $\mathbbm{E}[\frac{dA(t)}{dt}] + \mathbbm{E}[\frac{d\Phi(t)}{dt}] \leq \frac{c}{\epsilon^2}\mathbbm{E} [\frac{dOPT(t)}{dt}]$.
\end{itemize}

By integrating over time, one can see that the existence of such a potential function is sufficient to yield a $(1+\epsilon)-speed$ $o(\frac{1}{\epsilon^2})-competitive$
algorithm. Refer to Appendix \ref{proof_a} for the detailed proof.

%\vspace{-.5em}
\section{Performance Evaluation}
\label{evaluation}
In this section, we evaluate the performance of the SRPTMS+C algorithm via
extensive simulations driven by Google cluster-usage traces
\cite{Google_trace}. The traces contain the information of job submission and completion time of
Google services on a cluster
of 12K servers. It also includes the number of tasks as well as the duration of each task. In addition, the priority for each job
ranges from 0 to 11 and we just treate this priority as the job weight.
From the traces, we extract the statistics of more than 6000 jobs during a 12-hour period.
 We already exclude those jobs which have specific constraints on machine attributes.
The detailed job statistics are illustrated in Table \ref{table-trace-1}.

%Table \ref{table-trace} describes the detailed job statistics for this trace.
\begin{table}
%\vspace{.5em}
\normalsize
\centering
\caption{Google trace data statistics}
\label{table-trace-1}
\begin{tabular}{|c|c|}
  \hline
  Total number of Jobs &  6064\\
  \hline
  Trace duration (s) & 35032\\
  \hline
  Average number of tasks per job & 26.31 \\
  \hline
  Minimum task duration (s) & 12.8 \\
  \hline
  Maximum task duration (s) & 22919.3 \\
  \hline
  Average task duration (s) & 1179.7 \\
  \hline
\end{tabular}
\vspace{-1em}
\end{table}

When running the simulations, we estimate the distribution for the workloads of all the tasks within each job phase. Once a cloning copy is made for a particular task, the workload for this clone is just drawn independently from the estimated distribution. We repeat the same simulation for each of the following evaluations ten times and take the
average to obtain the final result.

\begin{figure*}
%\vspace{.2em}
\centering
\begin{minipage}{.32\textwidth}
\centering
\includegraphics[width=\linewidth]{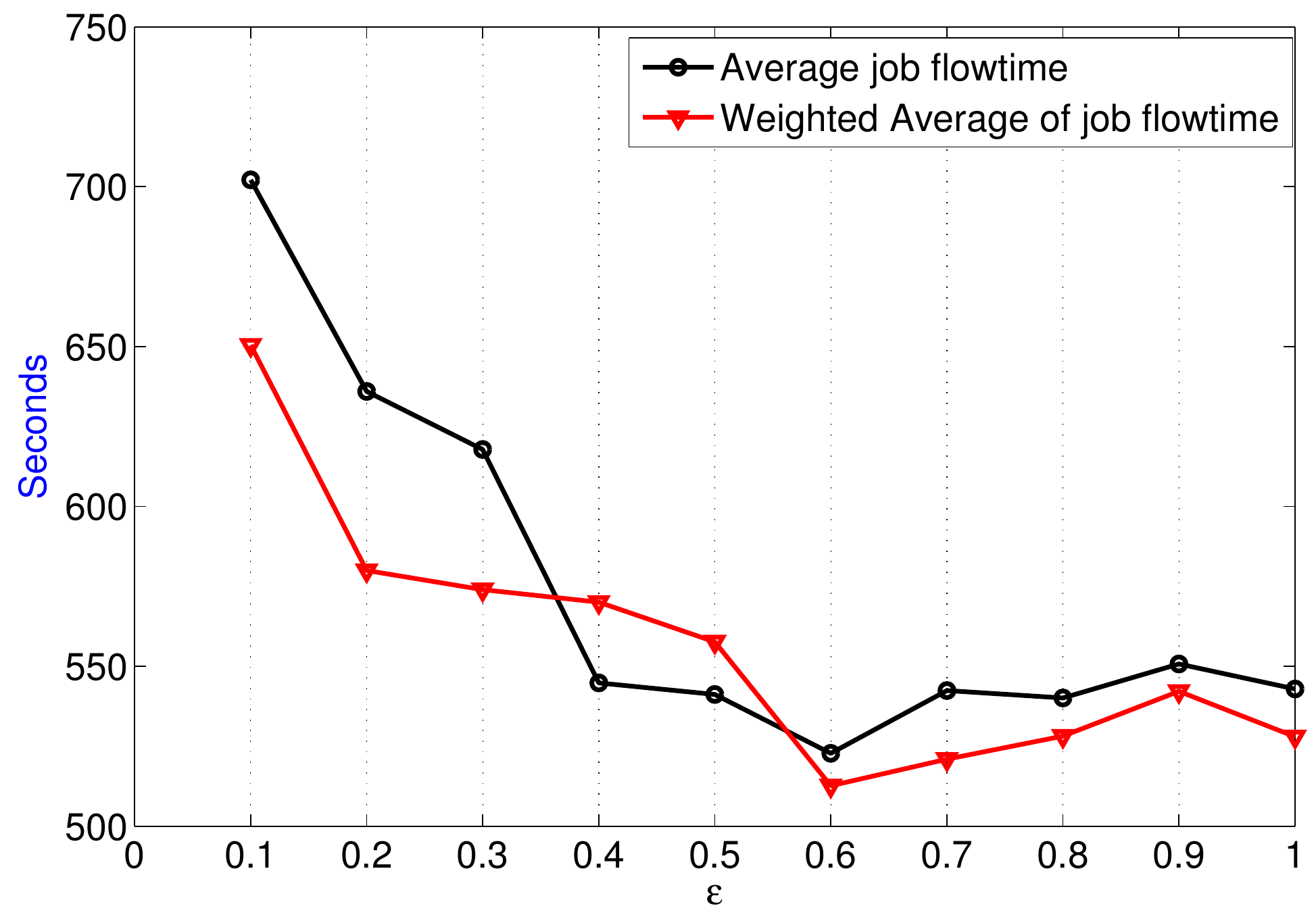}
\caption{The weighted/unweighted average of job flowtimes for different $\epsilon$ under the SRPTMS+C algorithm when $r = 0$.}
\label{epsilon}
\end{minipage}\hfill
\begin{minipage}{.32\textwidth}
\centering
\includegraphics[width=\linewidth]{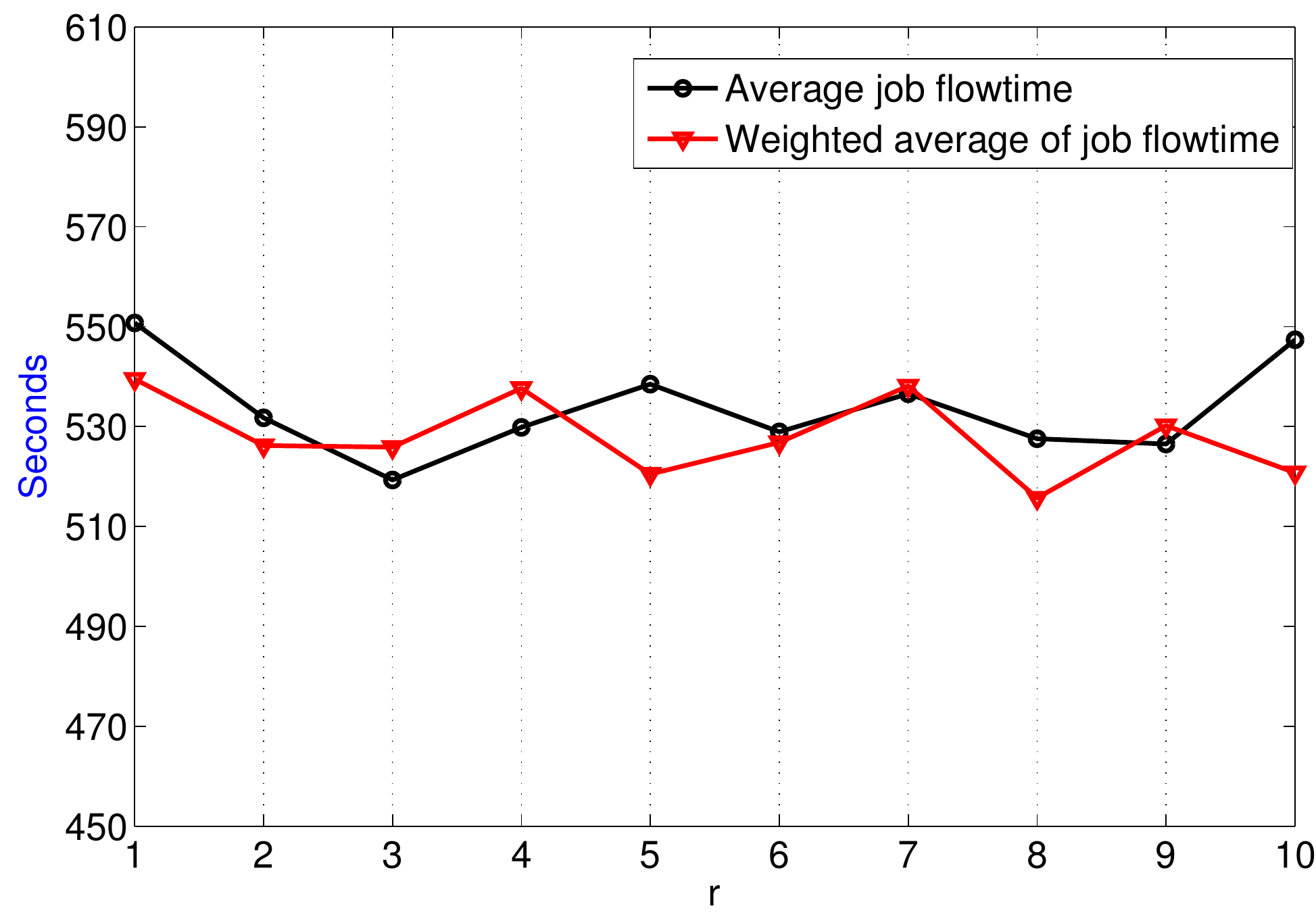}
\caption{The weighted/unweighted average of job flowtimes for different $r$ under the SRPTMS+C algorithm when $\epsilon = 0.6$.}
\label{r}
\end{minipage}\hfill
\begin{minipage}{.32\textwidth}
\centering
\includegraphics[width=\linewidth]{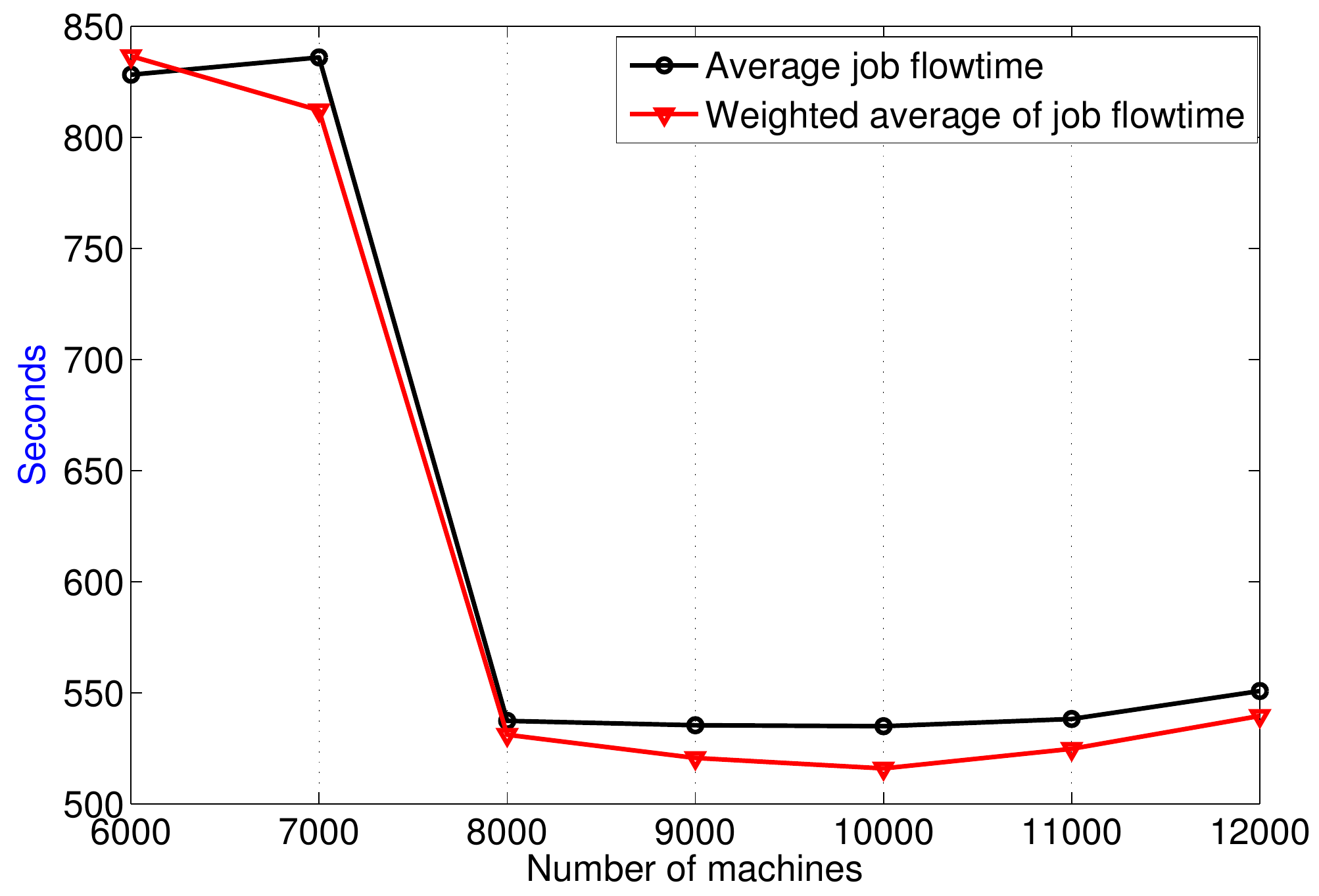}
\caption{The weighted/unweighted average of job flowtimes under different number of machines for SRPTMS+C when $\epsilon = 0.6$ and $r = 3$.}
\label{num}
\end{minipage}
\vspace{-0.8em}
\end{figure*}

\subsection{Baseline Algorithms for comparison}
We adopt the following two algorithms as the baselines to compare with the SRPTMS+C algorithm:

%\textbf{Baseline:}
\begin{itemize}
\item \textbf{Microsoft Mantri's Speculative Execution Scheme:} The speculative execution scheme of Mantri is demonstrated to be the most effective one among all the straggler-detection based schemes \cite{Outliers}. Mantri estimates  the remaining time to finish, $t_{rem}$, for each task and predicts the required execution time of a relaunched copy of the task, $t_{new}$. Once a machine becomes available, the system makes a decision on whether to launch a backup copy based on the statistics of $t_{rem}$ and $t_{new}$. Specifically, another copy is launched if the inequality $\mathbb{P}(t_{rem}>2*t_{new}) > \delta$ holds.
\item \textbf{Smart Cloning Algorithm (SCA):} SCA is a cloning algorithm which is proposed in \cite{speculative-multiple-optimization}. At the beginning of each time slot, SCA first runs a convex program to determine the number of copies assigned for each task and then launch all the copies simultaneously on available machines.  SCA has been demonstrated to cut
down the elapsed time of small jobs substantially.%\item \textbf{Enhanced Speculative Execution Algorithm (ESE):} Rather than making clones in a greedy manner, ESE relaunches a new copy for the running task only when a straggler is detected and  $t_{rem}>\sigma * \mathbbm{E}[t_{new}]$ is satisfied \cite{speculative-multiple-optimization}. Similar to Mantri scheme, ESE needs to monitor the progress of each task which incurs extra system instrumentation and performance overhead.
\end{itemize}

Instead of comparing the weighted sum of job flowtimes directly, we take the weighted average for ease of presentation. Moreover, we also compare the unweighted average as well as the cumulative distribution function (i.e., CDF) of job flowtimes against different algorithms. The time scale of each slot is 1 second in our simulations.

\begin{figure*}
\centering
\begin{minipage}{.32\textwidth}
\centering
\includegraphics[width=\linewidth]{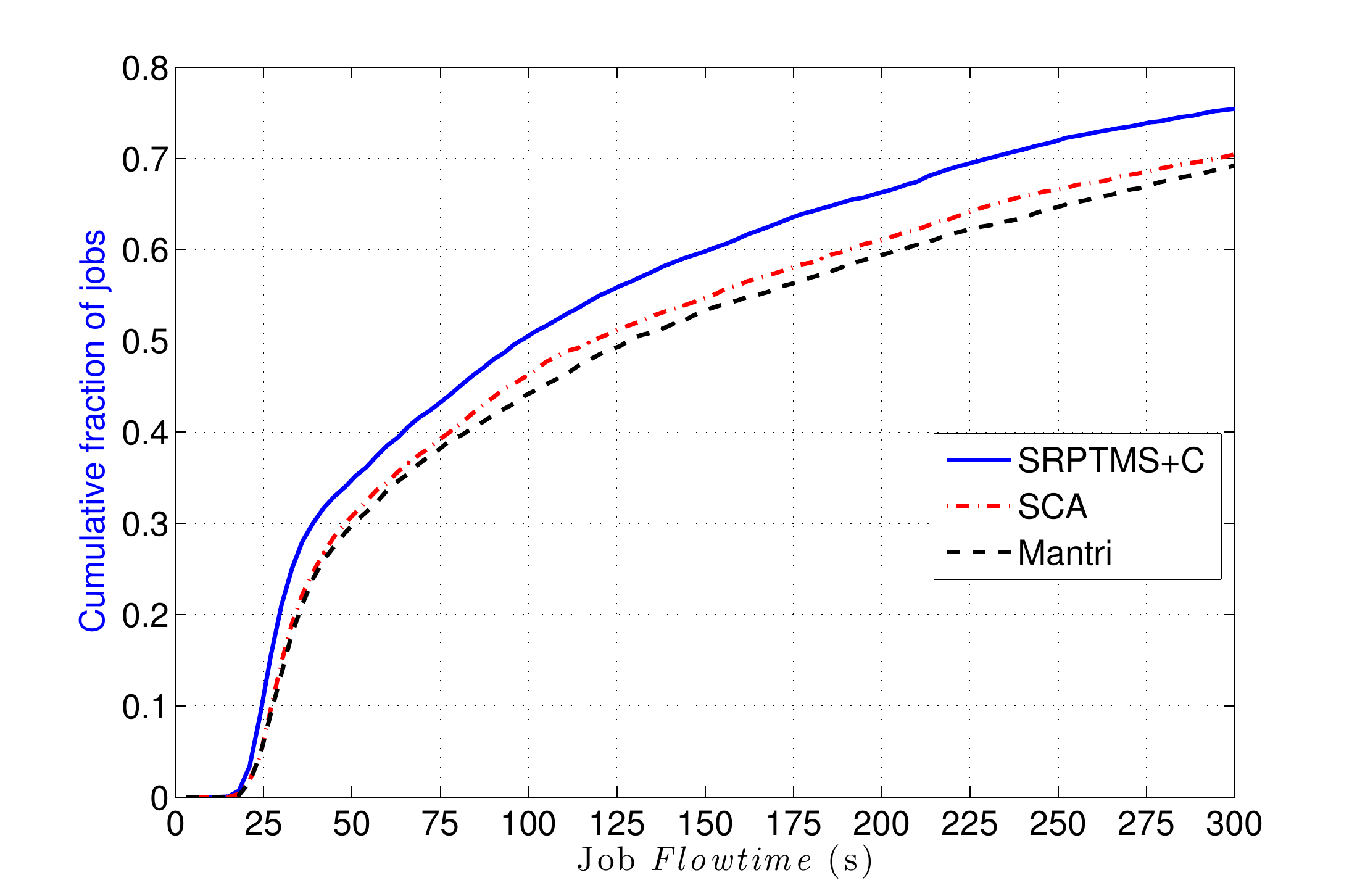}
\caption{The cumulative fraction of the jobs within
the  flowtime ranging from 0 to 300 seconds under different algorithms.}
\label{big_cmf}
\end{minipage}\hfill
\begin{minipage}{.32\textwidth}
\centering
\includegraphics[width=\linewidth]{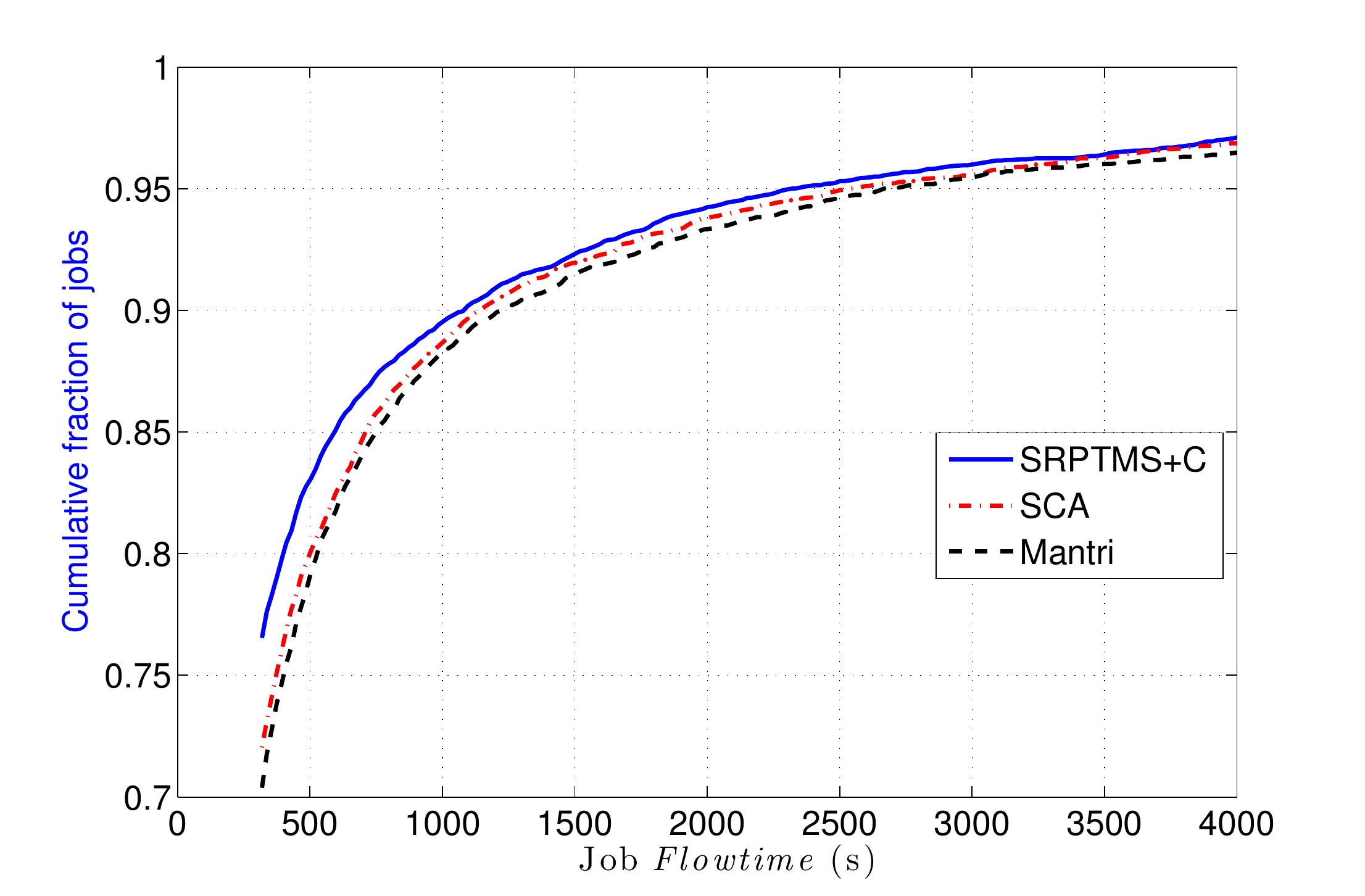}
\caption{The cumulative fraction of the jobs within
the flowtime ranging from 500 to 4000 seconds under different algorithms.}
\label{small_cmf}
\end{minipage}\hfill
\begin{minipage}{.32\textwidth}
\centering
\includegraphics[width=\linewidth]{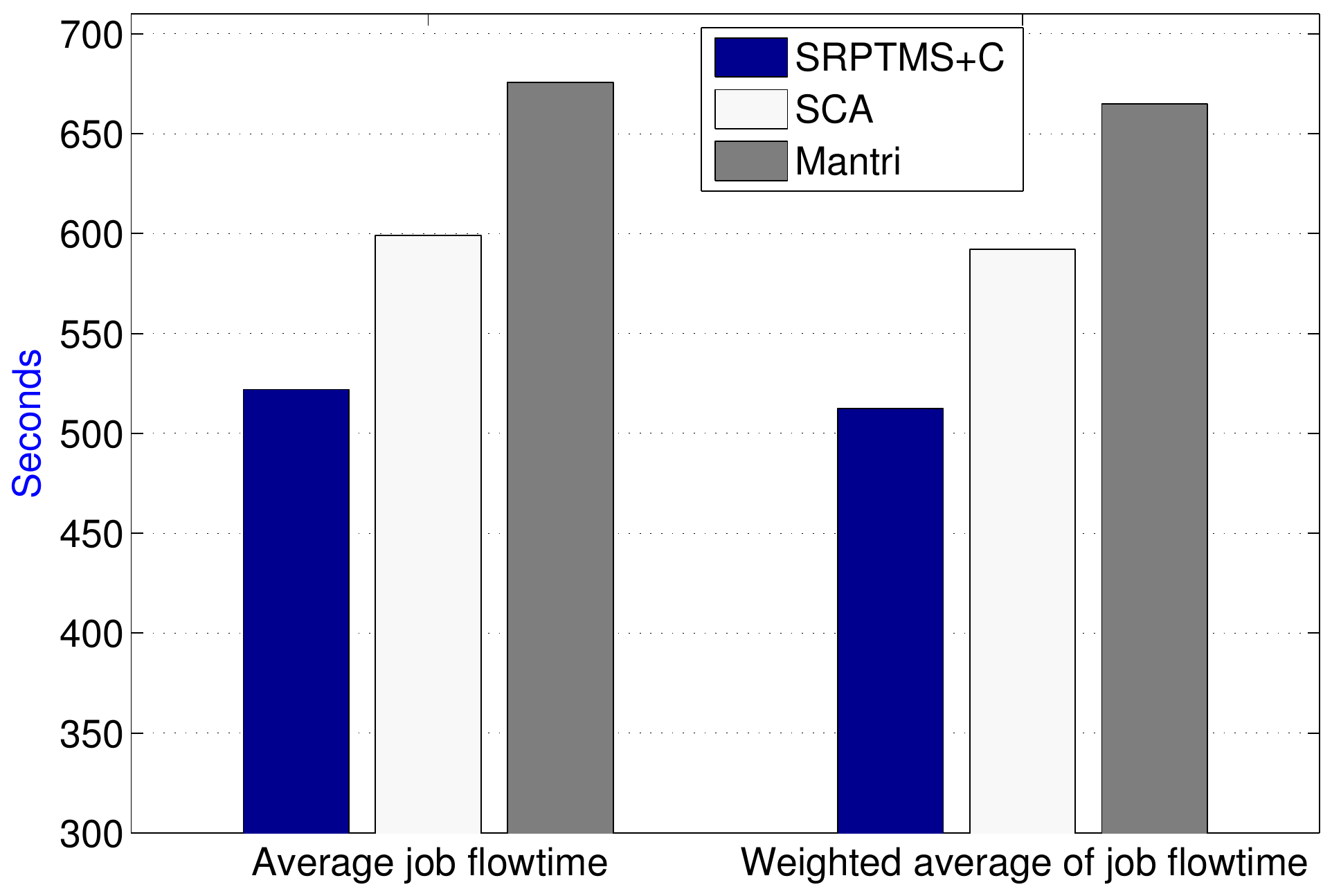}
\caption{The weighted/unweighted average of job flowtimes under different algorithms within the cluster that has 12K machines.}
\label{averagedflowtime}
\end{minipage}
\vspace{-1.0em}
\end{figure*}

\subsection{The impact of $\epsilon$, $r$ and the number of machines in the cluster}
In this subsection, we first evaluate the impact of $\epsilon$ and $r$ on the average  weighted/unweighted job flowtime under the SRPTMS+C algorithm in the cluster that
contains 12K machines.  Fig.~\ref{epsilon} depicts the evaluation result under different $\epsilon$ when $r=0$. Observe that when $\epsilon = 0.6$, which corresponding to the scheduler that schedules nearly half of the alive jobs with smaller effective workloads in each time slot, both of these two metrics attain the minimum. %Moreover, when $\epsilon$ is between 0.5 and 1, the weighted/unweighted average of job flowtimes does not vary much.

To further evaluate the impact of $r$ on the cluster performance, we set $\epsilon$ to 0.6 and evaluate the weighted/unweighted average of job flowtimes for different $r$ under SRPTMS+C algorithm. It shows in Fig.~\ref{r} that the unweighted average of job flowtimes attain the minimum when $r=3$ while the weighted average reaches its minimum under $r=8$.
In fact, both of these two metrics do not vary much between different $r$, the major reason is that the variation of task duration within each job phase for this particular job trace is small.

On the other hand, we scale out different number of machines in the cluster to show the impact on the job flowtime. Observe from Fig.~\ref{num}
that when the number of machines is around 8K, the performance is as equally well as  it in the original cluster with 12K machines.  There is  enough resources to
make clones for small jobs under the SRPTMS+C algorithm although the cluster only has 8K machines. The flowtime of small jobs from this trace therefore reduces substantially under SRPTMS+C.

\subsection{Comparison against baseline algorithms}
Based on the evaluation results above, we choose $\epsilon$ to be 0.6 and $r$ to be 3 for the SRPTMS+C algorithm. We implement the three baseline algorithms as presented in their original papers in the cluster that contains 12K machines. The comparison results are illustrated in Fig.~\ref{big_cmf} and Fig.~\ref{small_cmf}. Fig.~\ref{big_cmf} depicts the CDF of  flowtime for the small jobs whose
 flowtime is between 0 and 300 seconds.  It indicates that the SRPTMS+C algorithm obtains the best performance for those small jobs. In SRPTMS+C, more than 50\% jobs complete within 100 seconds. In contrast, about 46\% and 44\% jobs complete within 100 seconds under SCA and Mantri respectively.

Fig.~\ref{small_cmf} depicts the CDF of  flowtime for the big jobs whose
 flowtime are between 300 and 4000 seconds. One can see that the SRPTMS+C algorithm still achieves the best performance for these big jobs. For instance, about 90\% jobs can complete within 1000 seconds under SRPTMS+C while only 88\% and 86\% jobs can complete within such time-span under SCA and Mantri respectively.

We illustrate the weighted/unweighted average of job flowtimes for this trace under different algorithms in Fig.~\ref{averagedflowtime}. It shows that both of these two metrics under SRPTMS+C are reduced by nearly 25\% comparing to Mantri baseline scheme. More importantly, the SRPTMS+C algorithm is much more efficient comparing to Mantri scheme in terms of implementation as the latter needs to monitor the progress of each running task which induces an extra system instrumentation.

\section{Conclusions}
\label{conclusion}
In this paper, we study the online scheduling problem in a  MapReduce cluster and formulate a stochastic optimization program with the objective to minimize the
weighted sum of job flowtimes. Following this model, we design the straggler-mitigation algorithms via task cloning,  which are motivated from the SRPT scheduler  in
both offline and online cases. In the offline case, we show that, with high probability, each job can complete within a time-span which is 
larger than the optimal scheduling algorithm by only a constant factor times the standard derivation of task duration under our algorithm. 
When the variance of task duration is negligible, the offline algorithm achieves a competitive ratio of 2. On the other hand,
we present the SRPTMS+C algorithm for the online case and provide an upper bound for the competitive ratio through the potential function analysis. Finally, we run several trace-driven simulations to evaluate the performance of the SRPTMS+C algorithm. It shows that SRPTMS+C cuts down the flowtime of small jobs substantially
 and reduces the wighted/unweighted sum of job flowtimes by nearly 25\% comparing to Mantri baseline scheme.

\bibliographystyle{abbrv}
\bibliography{scheduling}

\appendix
%\subsection{Algorithm }
\subsection{Proof of Lemma \ref{Lemma_1}}
\label{proof_lemma_1}
\begin{proof}
Consider the following case where the cluster is processing some jobs with priorities smaller than $w_i/\phi_i$ during the interval $[0,f_i-E^r_i - r\sigma^r_i]$. In this case, the job $J_i$ must have already scheduled all the reduce tasks at time $f_i-E^r_i - r\sigma^r_i$ according to the offline algorithm we present. Further, we consider the reduce task which is finished at last in job $J_i$ and let $\delta^{r,j}_i$ denote it. According to the definition of $f_i$, $\delta^{r,j}_i$ finishes its work at time $f_i$. It implies that  the workload of $\delta^{r,j}_i$ is at least $E^r_i + r\sigma^r_i$. Applying the Chebyshev Inequality \cite{probability} here, we get the following formula:
\begin{equation}
\Pr\{r^j_i \geq E^r_i + r\sigma^r_i\} \leq \Pr\{|r^j_i - E^r_i| \geq r\sigma^r_i\} \leq \frac{1}{r^2}
\end{equation}
This completes the proof.
\end{proof}

\subsection{Proof of Theorem \ref{Theorem 1}}
\label{proof_theorem_1}
\begin{proof}
Denote by $X$ the work that the cluster has processed for the jobs with priority at least $w_i/\phi_i$. Thus, the following two equations hold:
\begin{equation}
E[X] = \sum_{j:w_j/\phi_j \geq w_i/\phi_i} m_j\cdot E^m_j  + r_j\cdot E^r_j.
\end{equation}

\begin{equation}
\sigma^2[X] = \sum_{j:w_j/\phi_j \geq w_i/\phi_i} m_j\cdot (\sigma^r_j)^2 + r_j\cdot (\sigma^r_j)^2.
\end{equation}
Applying the Chebyshev Inequality again, we conclude that the probability that $X$ is no less than $f_i^s$ is bounded by

\begin{equation}
\begin{split}
& \Pr \left\{X \geq f_i^s \right\} \\ & \leq \Pr \left\{\left|X - E[X]\right| \geq r\sum_{j:w_j/\phi_j \geq w_i/\phi_i}(m_j \sigma^m_j+r_j \sigma^r_j) \right\} \\ & \leq \frac{\sum_{j:w_j/\phi_j \geq w_i/\phi_i} m_j\cdot (\sigma^r_j)^2 + r_j\cdot (\sigma^r_j)^2}{[r\sum_{j:w_j/\phi_j \geq w_i/\phi_i}(m_j\cdot \sigma^m_j+r_j\cdot \sigma^r_j)]^2} \\ & \leq \frac{1}{r^2}
\end{split}
\end{equation}

Applying Lemma \ref{Lemma_1} here, with probability at least $\frac{r^2-1}{r^2}$, the cluster is processing the work of $X$ during the interval $[0,f_i-E^r_i - r\sigma^r_i]$.
There are M machines processing the task with unit speed in total. Thus, with probability at least $(r^2-1)^2/r^4$, the following inequality holds.

\begin{equation}
M*(f_i-E^r_i -r\sigma^r_i) \leq f_i^s
\end{equation}
The theorem immediately follows. Q.E.D.
\end{proof}

\subsection{Proof of Theorem \ref{competitive_ratio}}
\label{proof_a}
In the algorithm design of Section \ref{algorithm_design}, we assume that time is slotted. For convenience of analysis, here we
consider a more general case where the time is continuous. In fact, we just make the length of
a time slot small enough, as long as the duration of each task is the multiples of
a slot length, our analysis doesn't violate the algorithm setting.

\begin{proof}
 Let $y^j_i (t) = \max \{ p_i^{A j} (t) - p_i^{O j} (t) , 0 \}$ where $p_i^{O j}
(t)$ and $p_i^{A j} (t)$ denote the remaining workload to be processed for task
$\delta^j_i$ in Job $J_i$ at time $t$ under the optimal scheduling policy and
SRPTMS+C algorithm respectively. Let $\psi^o (t)$ and $\psi^o_s (t)$ be the jobs and
tasks that are still alive (have not completed yet) at time $t$ in the optimal scheduling. Further
denote by $\psi^s (t)$ the set of jobs that are alive in SRPTMS+C.

Denote by $t^{s,j}_i$ and $t^{f,j}_i$ the start time and completion time of task $\delta^j_i$ respectively. Based on the constraints \eqref{mapper_processing} and \eqref{reducer_processing}, it follows that
\begin{equation}
\label{average_duration}
\mathbbm{E} \left[ t^{f,j}_i - t^{s,j}_i \right] = \mathbbm{E} \left[ t_i^{j} \right]/s_i(x_i^{j})
\end{equation}
moreover, we have
\begin{equation}
\label{average_processing}
\mathbbm{E} \left[\int_{t^{f,j}_i}^{t^{s,j}_i} d p_i^{A j} (t) \right] = \mathbbm{E} \left[ t_i^{j} \right]
\end{equation}
Substituting  Equation \eqref{average_duration} into Equation \eqref{average_processing} yields
the following formula:
\begin{equation}
\label{average_speed}
\mathbbm{E} \left[\frac{d p_i^{A j} (t)}{dt} \right] = -s_i(x_i^{j})
\end{equation}

The potential function for a single task is defined as follows:
\begin{equation}
  \varphi_i^j (t) = \frac{w_i y_i^j (t)}{s_i (w_i M / \varepsilon W (t))}
\end{equation}
Our overall potential function for all the jobs in the cluster is defined as
\begin{equation}
  \Psi (t) = \frac{1}{\varepsilon^2} \sum_{i \text{} \in \psi^s (t)}
  \sum_{\delta^j_i \in J_i^c (t)} \varphi_i^j (t)
\end{equation}
The Potential funciton is differentiable and it holds that $\Psi (0) = \Psi
(\infty) = 0$ and
\begin{equation}
  \mathbbm{E} \left[ \frac{d \Psi (t)}{\tmop{dt}} \right] =
  \frac{1}{\varepsilon^2} \sum_{i \text{} \in \psi^s (t)} \sum_{\delta^j_i \in
  J_i^c (t)} \mathbbm{E} \left[ \frac{d \varphi_i^j (t)}{\tmop{dt}} \right]
\end{equation}

Let $C_i^A$ be the completion time for Job $J_i$ under SRPTMS+C algorithm.  For any time
$t \geqslant r_i$, let $A_i (t) = w_i (\min \{ C_i^A, t \} - r_i)$ be the
accumulated weighted flow time of Job $J_i$ at time $t$, then we must have
\begin{equation}
  \mathbbm{E} \left[ \frac{\tmop{dA}_i (t)}{\tmop{dt}} \right] = w_{_i} \quad for \quad r_i < t < C_i^A
\end{equation}
$A_i (\infty)$ is just the flowtime  of Job $J_i$ in the cluster. Hence,
the total flowtime of all the jobs in the cluster can be formulated as $A =
\sum_i A_i (\infty)$. Futher let $A_i = A_i (C_i^A)$ and $A(t) = \sum_{i}A_i(t)$. Let $\tmop{OPT}_i (t)$,
$\tmop{OPT}_i$, $\tmop{OPT} (t)$ and OPT be defined similarly for the optimal
scheduling policy.

Similar to the potential-function based analysis \cite{competitive,energy_efficient,scalably-scheduling,SRPT_identical}, our goal is to
bound the continuous and discrete increases to $\Psi (t)$ by a function of
OPT.

We now focus on the the changes made to $\Psi (t)$. It's obvious that the job
arrivals make no change to this metric. In addition, the completion of jobs in
the optimal schedule has no effect on the potential function value. The completion
of jobs in SRPTMS+C causes the corresponding term being removed from $\Psi
(t)$, however, it only decreases the potential and we just omit it as our goal
is to obtain the upper bound for the changes made to $\Psi (t)$. As a result,
we only need to analyze the continuous change to $\Psi (t)$.

\tmfolded{Changes in $\Psi (t)$ due to the optimal scheduling policy which is
define as $\Delta^O (t)$: }{\ }

Let $a_i^{\tmop{Oj}}$ be the number of machines assigned to task $\delta_i^j$
of Job $J_i$ in the optimal scheduling policy. Based on Equation \eqref{average_speed} and the definition of
 potential function, the contribution made by the
optimal scheduling to $\frac{d}{\tmop{dt}} \mathbbm{E} [\varphi_i^j (t)]$ is
bounded by the following formula:
\begin{equation}
  \varepsilon^2 \Delta_i^{\tmop{Oj}} = \frac{w_i s_{_i} (a_i^{\tmop{Oj}})}{s_i
  (w_i M / \varepsilon W (t))}
\end{equation}
There are two categories for $a_{i^{}}^{\tmop{Oj}}$ which are
$a_{i^{}}^{\tmop{Oj}} \leqslant w_i M / \varepsilon W (t)$ and
$a_{i^{}}^{\tmop{Oj}} > w_i M / \varepsilon W (t)$. For the former case,
appling the monotonic property of $s_i$ function for all $i$, we have
$\frac{w_i s_{_i} (a_i^{\tmop{Oj}})}{s_i (w_i M / \varepsilon W (t))}
\leqslant w_i$. For the latter case, applying Proposition 1, we get
\begin{equation}
  \frac{w_i s_{_i} (a_i^{\tmop{Oj}})}{s_i (w_i M / \varepsilon W (t))}
  \leqslant w_i \frac{}{} \frac{a_i^{\tmop{Oj}}}{w_i M / \varepsilon W (t)} =
  \varepsilon W (t) \frac{a_i^{\tmop{Oj}}}{M}
\end{equation}
Combining the two cases, it follows that
\begin{eqnarray}
  \varepsilon^2 \Delta_i^{\tmop{Oj}} (t) & \leqslant & \max \left\{ w_i,
  \varepsilon W (t) \frac{a_i^{\tmop{Oj}}}{M} \right\} \\
   & \leqslant & w_i + \varepsilon W (t) \frac{a_i^{\tmop{Oj}}}{M}
\end{eqnarray}
which indicates
\begin{eqnarray}
  \Delta_{}^O (t) & = & \sum_{i \text{} \in \psi^o (t) \cap \psi^s (t)}
  \sum_{\delta^j_i \in J_i^c (t)} \Delta_i^{\tmop{Oj}} (t) \\
   \label{machine_optimal}
  & \leqslant & \frac{1}{\varepsilon^2} \sum_{i \text{} \in \psi^o (t) \cap
  \psi^s (t)} \sum_{\delta^j_i \in J_i^c (t)} w_i \nonumber\\
  &  & + \frac{W (t)}{M \varepsilon^{}} \sum_{i \text{} \in \psi^s (t)}
  \sum_{\delta^j_i \in J_i^c (t)} a_i^{\tmop{Oj}}
\end{eqnarray}
For the first term of Equation \eqref{machine_optimal}, it follows that $\sum_{\delta^j_i \in
J_i^c (t)} w_i \leqslant \tmop{C w}_i$ where $C$ is the maximum number of copies made for each task in the
optimal scheduling algorithm. Hence,
\begin{eqnarray}
\label{optimal_dynamic}
  \frac{1}{\varepsilon^2} \sum_{i \text{} \in \psi^o (t) \cap \psi^s (t)}
  \sum_{\delta^j_i \in J_i^c (t)} w_i & \leqslant & \frac{C}{\varepsilon^2}
  \sum_{i \text{} \in \psi^o (t) \cap \psi^s (t)} w_i \nonumber\\
  & \leqslant & \frac{C}{\varepsilon^2} \sum_{i \text{} \in \psi^o (t)} w_i
  \nonumber\\
  & = &  \frac{C}{\varepsilon^2} \cdot \mathbbm{E} \left[ \frac{d \tmop{OPT}
  (t)}{\tmop{dt}} \right]
\end{eqnarray}
For the second term in Equation \eqref{machine_optimal}, we have
\begin{equation}
\label{total_number}
  \sum_{i \text{} \in \psi^s (t)} \sum_{\delta^j_i \in J_i^c (t)}
  a_i^{\tmop{Oj}} \leqslant M
\end{equation}
and
\begin{equation}
\label{total_weight}
  W (t) = \sum_{i \text{} \in \psi^s (t)} w_i = \sum_{i \text{} \in \psi^s
  (t)} \mathbbm{E} \left[ \frac{\tmop{dA}_i (t)}{\tmop{dt}} \right]
  =\mathbbm{E} \left[ \frac{d A (t)}{\tmop{dt}} \right]
\end{equation}
Substitute Equation \eqref{optimal_dynamic}, \eqref{total_number} and \eqref{total_weight} into Equation \eqref{machine_optimal}, it yields that
\begin{equation}
  \begin{array}{lll}
    \Delta_{}^O (t) & \leqslant & \frac{C}{\varepsilon^2} \cdot \mathbbm{E}
    \left[ \frac{d \tmop{OPT} (t)}{\tmop{dt}} \right] + \frac{1}{\varepsilon}
    \mathbbm{E} \left[ \frac{d A (t)}{\tmop{dt}} \right]^{}
  \end{array}
\end{equation}
We proceed to analyze the changes to $\Psi (t)$ made by our SRPTMS+C
scheduling.

\tmfolded{Changes in $\Psi (t)$ due to the SRPTMS+C scheduling policy which
is defined as $\Delta_{}^S (t)$:}{\ }

For \ each task that is alive in SRPTMS+C at time $t$, if it completes the
work in the optimal scheduling policy, then $y_i^j (t)$ is positive. Hence,
$y^j_i (t)$ decreases for all tasks $\delta^{^j}_i \nin \psi^o_s (t)$ that
SRPTMS+C processes at time $t$.

We run our algorithm at speed of $1 + \varepsilon$. Let $a^{\tmop{Sj}}_i$ be
the number of machines assigned to task $\delta^j_i$ in SRPTMS+C at time $t$.
According to our scheduling policy, \ we have $\sum_j a_i^{\tmop{Aj}}
\leqslant g_i (t)$ for all $J_i \in \psi^s (t)$. It follows that
\begin{eqnarray}
\label{con_srpt}
  \Delta_{}^S (t) & \leqslant & - \frac{1 + \varepsilon}{\varepsilon^2}
  \sum_{i \text{} \in \psi^s (t)} \sum_{\delta^j_i \in J_i^c (t) \cap
  \delta^{^j}_i \nin \psi^o_s (t)} \frac{w_i s_{_i} (a_i^{\tmop{Sj}})}{s_i
  (w_i M / \varepsilon W (t))} \nonumber\\
  & \leqslant & - \left( \frac{1 + \varepsilon}{\varepsilon^{}} \right) W (t)
  \sum_{i \text{} \in \psi^s (t)} \sum_{\delta^j_i \in J_i^c (t) \cap
  \delta^{^j}_i \nin \psi^o_s (t)} \frac{a_i^{\tmop{Sj}}}{M} \nonumber\\
  & = & \frac{- (1 + \varepsilon) W (t)}{\varepsilon^{}} \left( \sum_{i
  \text{} \in \psi^s (t) \atop \scriptstyle \delta^j_i \in J_i^c (t)} \frac{a_i^{\tmop{Sj}}}{M}
  - \sum_{\delta^{^j}_i \in \psi^o_s (t)} \frac{a_i^{\tmop{Sj}}}{M} \right)
  \nonumber\\
  & = & - \left( \frac{1 + \varepsilon}{\varepsilon^{}} \right) W (t) \sum_{i
  \text{} \in \psi^s (t)} \frac{g_i (t)}{M} \nonumber\\
  &  & + \left( \frac{1 + \varepsilon}{\varepsilon^{}} \right) W (t)
  \sum_{\delta^{^j}_i \in \psi^o_s (t)} \frac{a_i^{\tmop{Sj}}}{M}
\end{eqnarray}
The second inequality in the above follows Proposition \ref{convex_function} and the fact that
$a_i^{\tmop{Sj}} \leqslant g_i (t) = \frac{w_i M}{\varepsilon W (t)}$. To
bound Inequality \eqref{con_srpt}, we need to bound the second term as follows:
\begin{eqnarray}
\label{optimal_srpt}
  \sum_{\delta^{^j}_i \in \psi^o_s (t)} a_i^{\tmop{Sj}} & \leqslant & \sum_{i
  \in \psi^o (t)} \frac{w_i M}{\varepsilon W (t)} \\
  \label{optimal_srpt_1}
  & \leqslant & \frac{M}{\varepsilon W (t)} \mathbbm{E} \left[ \frac{d
  \tmop{OPT} (t)}{\tmop{dt}} \right]
\end{eqnarray}
It addition, we have
\begin{equation}
\label{total_number_machine}
  \sum_{i \text{} \in \psi^s (t)} g_i (t) = M
\end{equation}
Substitute Inequality \eqref{optimal_srpt}, \eqref{optimal_srpt_1} and Equation \eqref{total_number_machine} into Inequality \eqref{con_srpt}, it holds
that
\begin{eqnarray}
  \Delta_{}^S (t) & \leqslant & - \left( \frac{1 +
  \varepsilon}{\varepsilon^{}} \right) W (t) + \left( \frac{1 +
  \varepsilon}{\varepsilon^2} \right) \mathbbm{E} \left[ \frac{d \tmop{OPT}
  (t)}{\tmop{dt}} \right] \nonumber\\
  & = & - \left( \frac{1 + \varepsilon}{\varepsilon^{}} \right) \mathbbm{E}
  \left[ \frac{d A (t)}{\tmop{dt}} \right] \nonumber\\
  &  & +^{} \left( \frac{1 + \varepsilon}{\varepsilon^2} \right) \mathbbm{E}
  \left[ \frac{d \tmop{OPT} (t)}{\tmop{dt}} \right]
\end{eqnarray}
Wr proceed to complete the final analysis based on the results derived
above. Due to the fact that $\int_0^{\infty} \mathbbm{E} \left[ \frac{d \Psi
(t)}{\tmop{dt}} \right] \tmop{dt} =\mathbbm{E} [\nobracket \Psi (\infty)]
-\mathbbm{E} [\nobracket \Psi (0)] = 0,$ we have
\begin{eqnarray}
  \mathbbm{E} [A] & = & \int_0^{\infty} \mathbbm{E} \left[ \frac{\tmop{dA}
  (t)}{\tmop{dt}} \right] \tmop{dt} + \int_0^{\infty} \mathbbm{E} \left[
  \frac{d \Psi (t)}{\tmop{dt}} \right] \tmop{dt} \nonumber\\
  & \leqslant & \int_0^{\infty} \mathbbm{E} \left[ \frac{\tmop{dA}
  (t)}{\tmop{dt}} \right] \tmop{dt} + \int_0^{\infty} (\Delta_{}^O (t) +
  \Delta_{}^S (t)) \tmop{dt} \nonumber\\
  & \leqslant & \int_0^{\infty} \mathbbm{E} \left[ \frac{\tmop{dA}
  (t)}{\tmop{dt}} \right] \tmop{dt} + \int_0^{\infty} \left( -\mathbbm{E} \left[ \frac{\tmop{dA}
  (t)}{\tmop{dt}} \right]  \right) \tmop{dt} \nonumber\\
  &  & + \int_0^{\infty} \left( ^{} \frac{C + 1 + \varepsilon}{\varepsilon^2}
  \mathbbm{E} \left[ \frac{d \tmop{OPT} (t)}{\tmop{dt}} \right] \right)
  \tmop{dt} \nonumber\\
  & = & \left( \frac{C + 1 + \varepsilon}{\varepsilon^2} \right) \mathbbm{E}
  [\tmop{OPT}]
\end{eqnarray}
This completes the proof.
\end{proof}

%----------------------------------------------------------------------
\end{document}